\newcommand\SANSCOMMENTAIRE[1]{}

\newcommand\shorter[2][]{}

\newcommand\shortermcu[1]{}

\newcommand\Lesfonctionsquoi{\mathcal{F}}
\newcommand\manonclass{\linearderivlength^\bullet}
\newcommand\MANONlim{ELim}
\newcommand\MANONlimd{E_{2}Lim}
\newcommand\manonclasslim{\bar{\manonclass}}
\newcommand\polynomialb{ $\bar{\fonction{cond}}$-polynomial}
\newcommand{\signb}[1]{\bar{\fonction{cond}}(#1)}
\newcommand\unaire[1]{with respect to the value of $#1$}
\newcommand\Encode{\textit{Encode}}
\newcommand\Decode{\textit{Decode}}

\newcommand\MYVEC[1]{\vec{#1}}

\newcommand\THEPAPIERS{\cite{MFCS2019,MFCSJournalProject}}

\documentclass[a4paper]{article}
\usepackage{mathptmx}
\usepackage[utf8]{inputenc}
\usepackage[T1]{fontenc}
\usepackage[final]{microtype}
\usepackage[english]{babel}
\usepackage{csquotes}
\usepackage{amsmath}
\usepackage{amsthm}
\usepackage{amsfonts}
\usepackage{commath}
\usepackage{bm}
\usepackage{xfrac}
\usepackage{mathtools}
\usepackage{authblk}
\usepackage{url}
\usepackage[pdfencoding=auto,psdextra,hidelinks]{hyperref}
\usepackage{bookmark}
\usepackage{float}
\usepackage{booktabs}
\usepackage{subcaption}
\usepackage{tikz}
\usetikzlibrary{calc,positioning,shapes,arrows}

\usepackage{color}
\usepackage{todonotes}

\usepackage{versions}
\excludeversion{pascettesection}

\newcommand\dyadic{\mathbb{D}}

 \newcommand\R{\mathbb{R}}
\newcommand\Q{\mathbb{Q}}
\newcommand\N{\mathbb{N}}
 \newcommand\Z{\mathbb{Z}}
 
 \newcommand\RR{\R}

\newcommand\Greg{Grzegorczyk}

\newcommand\BRN{{\rm BRN}}

\newcommand\ODE{{\rm ODE}}

\newcommand\LI{{\rm LI}}
\newcommand\Elem{\mathcal{E}}
\newcommand\Gregn{\mathcal{E}_n}
\newcommand\PR{\mathcal{P}\mathcal{R}}

\newcommand{\fonction}[1]{\textrm{#1}}

\newcommand\projection[2]{\mathbf{\pi}_{#1}^{#2}}
\newcommand{\sucs}{\mathbf{s}}

\newcommand\plus{\mathbf{+}}
\newcommand\minus{\mathbf{-}}
\newcommand\gE{\mathbf{E}}

\newcommand\dint[4]{\int_{#1}^{#2}{#3}{\delta #4}}

\newcommand\fallingexp[1]{\overline{2}^{#1}}

\newcommand{\succun}[1]{\mathbf{1}({#1})}
\newcommand{\succzero}[1]{\mathbf{0}({#1})}

\newcommand{\zero}{\mathbf{0}}

\newcommand{\sign}[1]{\fonction{sg}(#1)}

\newcommand{\cond}[3]{\fonction{if}(#1,#2,#3)}

\newcommand{\tu}[1]{\mathbf{#1}}

\newcommand{\cp}[1]{\mathbf{#1}}
\newcommand{\Ptime}{\cp{PTIME}} 
\newcommand{\classP}{\cp{P}}      

\newcommand{\NP}{\cp{NP}}

\newcommand{\FPtime}{\cp{FPTIME}}
\newcommand{\FP}{\cp{FP}}

\newcommand{\Pspace}{\cp{PSPACE}}

\newcommand{\FPspaceN}{\cp{FPSPACE}}

\newcommand{\suffix}{\textsf{suffix}}

\renewcommand\bar[1]{\overline{#1}}

\usepackage{color}

\usepackage{todonotes}

\newcommand{\olivier}[2][]{\SANSCOMMENTAIRE{\todo[inline,color=blue!40,caption={2do}, #1]{\begin{minipage}{\textwidth-4pt}
Olivier:			#2\end{minipage}}}}

\newcommand{\manon}[2][]{\SANSCOMMENTAIRE{\todo[inline,color=violet!20!white,caption={2do},#1]{\begin{minipage}{\textwidth-4pt}
Manon:			#2\end{minipage}}}}

\newcommand{\olivierpourmanon}[2][]{\SANSCOMMENTAIRE{\todo[inline,color=blue!40,caption={2do}, #1]{\begin{minipage}{\textwidth-4pt}
Olivier $\to$ Manon:			#2\end{minipage}}}}
\newcommand{\olivierpourlecteur}[2][]{\SANSCOMMENTAIRE{\todo[inline,color=blue!40,caption={2do}, #1]{\begin{minipage}{\textwidth-4pt}
Olivier $\to$ Arnaud:			#2\end{minipage}}}}

\newcommand{\dderiv}[2]{\frac{\partial #1}{\partial #2}}
\newcommand{\dderivL}[1]{\frac{\partial #1}{\partial \lengt}}
\newcommand{\dderivl}[1]{\frac{\partial #1}{\partial \ell}}

\newcommand{\lengt}{\mathcal{L}}
 \newcommand\lengthnotation{\ell}
\newcommand{\length}[1]{\mathrm{\lengthnotation}(#1)}
\newcommand{\degre}[1]{\mathrm{deg}(#1)}

\newcommand{\derivlength}{\mathbb{DL}}
\newcommand{\linearderivlength}{\mathbb{LDL}}

\newtheorem{theorem}{Theorem}
\newtheorem{proposition}{Proposition}
\newtheorem{lemma}{Lemma}

\newtheorem{definition}{Definition}
\newtheorem{remark}{NB}

\usepackage[
backend=biber,
style=alphabetic,
useprefix=false,
url=false,
doi=true,
eprint=false,
maxbibnames=99,
]{biblatex}
\bibliography{bournez,perso}

\begin{document}

\title{A characterization of polynomial time computable functions from the integers to the reals using discrete ordinary differential equations\thanks{This work has been partially supported by ANR Project $\partial IFFERENCE$.}}
\author{Manon Blanc \\ 
	manon.blanc@lix.polytechnique.fr \\
	École Polytechnique, LIX, Palaiseau, France, ENS Paris-Saclay, Gif-sur-Yvette, France
	 \and Olivier Bournez \\
 		olivier.bournez@lix.polytechnique.fr\\
 		École Polytechnique, LIX, Palaiseau, France}
 	\date{ }
\maketitle


\begin{abstract}
In a recent article, the class of functions from the integers to the integers computable in polynomial time has been characterized using discrete ordinary differential equations (ODE), also known as finite differences. Doing so, we pointed out the 
fundamental role of linear (discrete) ODEs
and classical ODE tools such as changes of variables to capture
computability and complexity measures, or as a tool for programming. 

In this article, we extend the approach to a characterization of functions from the integers to the reals computable in polynomial time in the sense of computable analysis. In particular, we provide a characterization of such functions in terms of the smallest class of functions that contains some basic functions, and that is closed by composition, linear length ODEs, and a natural effective limit schema.
\end{abstract}

\section{Introduction}

Ordinary differential equations is a natural tool for modeling many phenomena in applied sciences, with a very abundant literature (see
e.g. \cite{Arn78,BR89,CL55}) and is rather well
understood under many aspects.  In a series of recent articles, they have been shown to also correspond to some natural computational model, with a nice computability and complexity theory: See  \cite{DBLP:journals/corr/BournezGP16} for a survey.

In a recent article \THEPAPIERS, we investigated their discrete counterpart, that are called discrete ODEs, also known as difference equations.  The basic principle is, for a function $\tu f(x)$  to consider its discrete derivative defined as $\Delta \tu f(x)= \tu f(x+1)-\tu
        f(x)$. We will intentionally  also write $\tu f^\prime(x)$ for
        $\Delta \tu f(x)$ to help to understand
	statements with respect to their classical continuous counterparts. 
This associated derivative notion, called  \emph{finite differences},  has been widely studied in numerical optimization for function approximation \cite{gelfand1963calcul} and  in \emph{discrete calculus} 
\cite{graham1989concrete,gleich2005finite,izadi2009discrete,urldiscretecalculuslau} for  combinatorial analysis.
   While the underlying computational content of finite differences theory is clear and has been pointed out many times, no fundamental connections with algorithms and complexity had been formally established before  \THEPAPIERS, where it was proved that many complexity and computability classes from computation theory can actually be characterized algebraically using discrete ODEs. Even if such results were initially motivated by helping to understand the relationships between analog
computations and classical discrete models of computation theory, the relation between the two is currently unclear.

In the context of algebraic classes of functions, a classical notation is the following: Call \emph{operation}  an operator that takes finitely many functions, and returns some new function defined from them. Then, $[f_{1}, f_{2}, \dots, f_{k}; op_{1}, op_{2},\dots,op_{\ell}]$, denotes the smallest set of functions containing functions $f_{1}, f_{2}, \dots, f_{k}$ that is closed under operations $op_{1}$, $op_{2}$, \dots $op_{\ell}$. 
Call \emph{discrete function} a function of type $ f: S_{1} \times \dots \times S_{d} \to S'_{1} \times \dots S'_{d'}$, where each $S_{i},S'_{i}$ is either $\N$, $\Z$.
Write $\FPtime$ for the class of functions computable in polynomial time. 
A  main result of \THEPAPIERS{} is the following ($\linearderivlength$ stands for linear derivation on length):

\begin{theorem}[\cite{MFCS2019}] \label{th:ptime characterization 2}
For discrete functions, we have 
	$$\linearderivlength= \FPtime$$
	where $\linearderivlength = [\mathbf{0},\mathbf{1},\projection{i}{k}, \length{x}, \plus, \minus, \times, \sign{x} \ ; composition, linear~length~ODE].$
 \end{theorem}
 	That is to say,  $\linearderivlength$ (and hence $\FPtime$ for discrete functions) is the smallest subset of  functions,
	 that contains   the constant functions $\mathbf{0}$ and $\mathbf{1}$, the projections
         $\projection{i}{k}$,  the length function  $\length{x}$ (that maps an integer to the length of its binary representation), 
         the addition function $x \plus y$, the subtraction function $x \minus y$, the multiplication function $x\times y$ (that we will also often denote $x\cdot y$), the sign function $\sign{x}$
	and closed under composition (when defined)  and linear length-ODE
        scheme: The linear length-ODE
        scheme
         basically (a formal definition is provided in  Definition \ref{def:linear lengt ODE})   corresponds in defining functions from linear ODEs with respect to derivation with respect to the length of the argument,
        that is to say of the form $\dderivl{\tu f(x,\tu y)} = 	\tu A [\tu f(x,\tu y), 
	x,\tu y]  \cdot \tu f(x,\tu y) 
           + \tu B [\tu f(x,\tu y), 
	x,\tu y ].
        $
 Here, in the above description, we use the notation $\dderivl{\tu f(x,\tu y)}$, which corresponds in derivation of $\tu f$ along the length function:  Given some function $\lengt:\N^{p+1} \rightarrow \Z$, and in particular for the case of where $\lengt(x,\tu y)=\ell(x)$, 
	\begin{equation}\label{lode}
	\dderivL{\tu f(x,\tu y)}= \dderiv{\tu f(x,\tu y)}{\lengt(x,\tu
          y)} = \tu h(\tu f(x,\tu y),x,\tu y),
	\end{equation}
is a formal synonym for
$ \tu f(x+1,\tu y)= \tu f(x,\tu y) + (\lengt(x+1,\tu y)-\lengt(x,\tu y)) \cdot
\tu h(\tu f(x,\tu y),x,\tu y).$

\begin{remark}
This concepts, introduced in  \THEPAPIERS, is motivated by the fact that the latter expression is similar to
classical formula for classical continuous ODEs:
$$\frac{\delta f(x,\tu y )}{\delta x} = \frac{\delta
  \lengt (x,\tu y) }{\delta x} \cdot \frac{\delta f(x,\tu
  y)}{\delta \lengt(x, \tu y)},$$
  and hence this is similar to a change of variable. Consequently, a linear length-ODE is basically a linear ODE over variable $t$, once the change of variable $t=\ell(x)$ is done. 
\end{remark}

%
%
%

In particular, writing as usual $B^{A}$ for functions from $A$ to $B$, we have: 

 \begin{theorem}[\cite{MFCS2019}] \label{th:ptime characterization 2}
	$\linearderivlength \cap \N^{\N}= \FPtime \cap \N^{\N}.$ 
 \end{theorem}
 
This provides a  characterization of $\FPtime$ for discrete functions that does not require to specify an
explicit bound in the recursion, in contrast to Cobham's work \cite{Cob65}, nor to assign
a specific role or type to variables, in contrast to safe recursion or ramification \cite{bs:impl,Lei-LCC94}. The characterization happens to be very simple 
using only natural notions from the world of ODE.   

Our purpose in this article is to extend this to more general classes of functions.
In particular, 
this makes sense to try to characterize polynomial time functions from the reals to the reals. We consider here computability and complexity over the reals in the most classical sense, that is to say, computable analysis (see e.g. \cite{Wei00}). Indeed, considering that $\N \subset \R$, most of the basic functions and operations in the above characterization (for example, $+$, $-$, \dots) have a clear meaning over the reals. %
One clear difficulty is that discrete ODEs are about discrete schemata, while we would like to talk about functions over the continuum.  We did not succeed to do so yet, but we propose here a substantial step towards this direction:  We provide a characterization of polynomial time computable functions \emph{from the integers to the reals} using discrete linear ODEs: considering linear ODEs is very natural in the context of ODEs.

%
%
%
%
%
%
%
%

To do so, we naturally go to talking  about algebra of functions more general than discrete functions, that is to say over more general space than $\N$ and $\Z$. This  introduces some subtleties, and difficulties, that we discuss in this article, with our various concepts, definitions and statements. 
%
%
 We hence basically  consider in this article functions of type $ f: S_{1} \times \dots \times S_{d} \to S_{0}$, where each $S_{i}$ is either $\N$, $\Z$ or $\Q$ or $\R$. 
 Or possibly vectorial functions whose components are of this type. We denote  $\Lesfonctionsquoi$ for the class of such functions. 
Clearly, we can consider $\N \subset \Z \subset \Q \subset \R$, but as functions may have different type of outputs, composition is an issue. We simply admit that
composition may not be defined in some cases. In other words, we consider that composition is a partial operator: for example, given $f: \N \to \R$ and $g: \R \to \R$, the composition of $g$ and $f$ is defined as expected, but $f$ cannot be composed with a function such as $h: \N \to \N$.


We then consider the class  $$\manonclass = [\mathbf{0},\mathbf{1},\projection{i}{k},   \length{x}, \plus, \minus, \times,\signb{x},\frac{x}{2};{composition, linear~length~ODE}]$$
of functions of $\Lesfonctionsquoi$. 
%
%
Here
\shortermcu{\begin{itemize}
 \item} $\ell: \N \to \N$ is the length function, mapping some integer to the length of its binary representation,  $\frac{x}{2}: \R \to \R$ is the function that divides by $2$, and all other basic functions are defined exactly as for $\linearderivlength$, but considered here as functions from the reals to reals. 
%
%
%
\shortermcu{
\item}  $\signb{x}: \R \to \R$  is some piecewise affine function that 
           takes value $1$ for $x>\frac34$ and $0$ for $x<\frac14$, and continuous piecewise affine: In particular, its restrictions to the integer is the function $\sign{x}$ considered in $\linearderivlength$. 
   %
%
%
\shortermcu{
	 \end{itemize}
	 }
%
%
%
%
We prove the following ($\|.\|$ stands for the sup-norm).

\begin{theorem}[Main theorem $1$] \label{th:main:one}
A function  $\tu f: \N^{d} \to \R^{d'}$ is computable in polynomial time if and only if there exists $\tilde{\tu f}:\N^{d+1} \to \R^{d'} \in \manonclass$ such that
for all $\tu m \in \N^{d}$, $n \in \N$,
$\|\tilde{\tu f}(\tu m,2^{n}) - \tu f(\tu m) \| \le 2^{-n}.$
\end{theorem}

\begin{proof}
Assume there exists $\tilde{\tu f}:\N^{d+1} \to \R^{d'} \in \manonclass$ such that
for all $\tu m \in \N^{d}$, $n \in \N$,
$\|\tilde{\tu f}(\tu m,2^{n}) - \tu f(m) \| \le 2^{-n}.$ 
From Proposition  \ref{prop:mcu:un}, we know that $\tilde{\tu f}$ is computable in polynomial time (in the binary length of its arguments).
Then $\tu f(m)$ is computable: indeed, given some integers $\tu m$ and $n$, we can approximate $\tu f(\tu m)$ at precision $2^{-n}$ as follows: Approximate $\tilde{\tu f}(\tu m,2^{n+1})$ at precision $2^{-(n+1)}$ by some rational $q$, and output $q$. We will then have
$$\|q-\tu f(\tu m)\| \le \|q-\tilde{\tu f}(\tu m,2^{n+1}) \| + \|\tilde{\tu f}(\tu m,2^{n+1})-\tu f(\tu m)\| \le 2^{-(n+1)} + 2^{-(n+1)} \le 2^{-n}.$$
All of this is done in polynomial time in $n$ and the size of $\tu m$, and hence we get that $\tu f$ is polynomial time computable from definitions.
\end{proof}
%

From the fact that we have the reverse direction in previous theorem, it is natural to consider the operation that maps $\tilde{\tu f}$ to $\tu f$. Namely, we introduce the operation $\MANONlim$ ($\MANONlim$ stands for Effective Limit):

\begin{definition}[Operation $\MANONlim$] Given $\tilde{\tu f}:\N^{d+1} \to \R^{d'} \in \manonclass$ such that
for all $\tu m \in \N^{d}$, $n \in \N$,
$\|\tilde{\tu f}(\tu m,2^{n}) - \tu f(\tu m) \| \le 2^{-n},$ then 
$\MANONlim(\tilde{\tu f})$ is the (clearly uniquely defined) corresponding function  $\tu f: \N^{d} \to \R^{d'}$.
\end{definition}

%


We obtain our main result, that provides a characterization of polynomial time computable functions for functions from the integers to the reals. 

\begin{theorem}[Main theorem $2$] \label{th:main:two} 
A function $\tu f: \N^{d} \to \R^{d'}$.
is  computable in polynomial time if and only if all its components belong to $\manonclasslim$, where :

$$\manonclasslim= [\mathbf{0},\mathbf{1},\projection{i}{k},   \length{x}, \plus, \minus, \times,  \signb{x}, \frac{x}{2}; {composition, linear~length~ODE};\MANONlim].$$
\end{theorem}

In particular:

\begin{theorem}[Main theorem $2$]\label{th:main:twop} 
$\manonclasslim \cap \R^{\N} = \FPtime \cap \R^{\N} $
\end{theorem}
	 
	 \shortermcu{
The rest of the paper is organized as follows:} In Section \ref{sec:discreteode}, we recall \shorter{some basic statements from }the theory of discrete ODEs. In Section \ref{sec:defanalysecalculable}, we recall  required concepts from computable analysis. In Section \ref{manonclassdansfptime}, we prove that functions from $\manonclass$ are polynomial time computable. 
Section \ref{fptimedansmanonclass} is proving a kind of reverse implication for functions over words. Then this is extended in Section \ref{sec:computablereal} to functions from integers to the reals, and we obtain a proof of Theorem \ref{th:main:one}. Section \ref{sec:main:two} then proves Theorems \ref{th:main:two} and \ref{th:main:twop} . Section \ref{sec:generalizations} is some generalizations of these results. Section \ref{sec:conclusion} discusses future work and difficulties to go to functions of $\R^{\R}$.


\paragraph{Related work.}  Various computability and complexity classes have been recently characterized using (classical) continuous ODEs: 
The most up-to-date survey is \cite{DBLP:journals/corr/BournezGP16}. Dealing with discrete ODEs is really different, as most of the constructions heavily rely on some closure properties of continuous ODEs  not true for discrete ODEs, in  particular because there is no chain rule formula for discrete derivation. 
The idea of considering discrete ODE as a model of computation is due to \THEPAPIERS. 
In a non-ODE centric point of view, we are characterizing some complexity classes using particular discrete schemata. Recursion schemes constitutes a major approach of computability theory and to some extent of complexity theory. 
The foundational characterization of $\FPtime$ due to Cobham \cite{Cob65}, and then others based on safe recursion \cite{bs:impl} or
ramification (\cite{LM93,Lei94}), or for other classes \cite{lm:pspace}, gave birth to the very vivid field of \textit{implicit complexity} at the
interplay of logic and theory of programming: See 
\cite{Clo95,clote2013boolean} for monographs.
Our ways of simulating Turing machines have some reminiscence of similar constructions used in other contexts such as Neural Networks \cite{SS95,LivreSiegelmann}. But with respect to all previous contexts, as far as we know, only a few papers have been devoted to characterizations of complexity, and even computability, classes in the sense of computable analysis. There have been some attempts using continuous ODEs \cite{BCGH07}, or the so-called $\R$-recursive functions \cite{DBLP:journals/corr/BournezGP16}. For discrete schemata, we only know \cite{brattka1996recursive}  and  \cite{ng2021recursion}, focusing on computability and not complexity.

%

\section{Some concepts from the theory of discrete ODEs}
\label{sec:discreteode}
	
In this section, we  recall some concepts and definitions from discrete ODEs, either well-known or established in \THEPAPIERS. 
%
%
\newcommand\polynomial{ \fonction{sg}-polynomial}
%
%
%
%
%
We  need to slightly extend the concept of \polynomial{}  expression from \THEPAPIERS{} to allow expressions with $\signb$ instead of $\sign$. 

\olivier{Peut-être besoin d'être plus fin sur la forme de ce qu'on autorise dans la définition d'après pour que ca marche parfaitement/au plus simple}
\begin{definition}[Extension of \THEPAPIERS]
A \polynomialb{}  expression $P(x_1,...,x_h)$ is an expression built-on
$+,-,\times$ (often denoted $\cdot$) and $\signb{}$ functions over a set of variables $V=\{x_1,...,x_h\}$ and integer constants.
The degree $\deg(x,P)$ of a term $x\in V$ in $P$ is defined inductively as follows:
\shortermcu{
\begin{itemize}
	\item} $\deg(x,x)=1$ and for  $x'\in V\cup \Z$ such that $x'\neq x$, $\deg(x,x')=0$;
	\shortermcu{\item}  $\deg(x,P+Q)=\max \{\deg(x,P),\deg(x,Q)\}$;
\shortermcu{\item}   $\deg(x,P\times Q)=\deg(x,P)+\deg(x,Q)$;
\shortermcu{\item}   $\deg(x,\sign{P})=0$.
\shortermcu{
\end{itemize}}
A \polynomialb{}  expression $P$  is \textit{essentially constant} in
$x$ if $\degre{x,P}=0$. 
\end{definition}

Compared to the classical notion of degree in polynomial expression,
all subterms that are within the scope of a sign (that is to say $\signb{}$)  function contributes
 $0$ to the degree. A vectorial function (resp. a matrix or a vector) is said to be a \polynomialb{} expression if all
its coordinates (resp. coefficients) are. 
It is said to be
\textit{essentially constant} if all its coefficients are.

\begin{definition}[\THEPAPIERS] \label{def:essentiallylinear}
A 
\polynomialb{} expression $\tu g(\tu f(x, \tu y), x,
\tu y)$ is \textit{essentially linear} in $\tu f(x, \tu y)$ if
it is of the form $
	\tu g(\tu f(x, \tu y), x,
	\tu y) =
	\tu A [\tu f(x,\tu y), 
	x,\tu y]  \cdot \tu f(x,\tu y) 
           + \tu B [\tu f(x,\tu y), 
	x,\tu y ] $
 where $\tu A$ and $\tu B$ are \polynomialb{} expressions essentially
constant in $\tu f(x, \tu y)$.
\end{definition}

For example, 
    the expression $P(x,y,z)=x\cdot \signb{(x^2-z)\cdot y} + y^3$
    is essentially linear in $x$, essentially constant in $z$ and not linear in
    $y$. 
    \shortermcu{
     \item 
      The expression
    $P(x,2^{\length{y}},z)=\signb{x^2 - z}\cdot z^2 + 2^{\length{y}}$
    is essentially constant in $x$, essentially linear in
    $2^{\length{y}}$ (but not essentially constant) and not
    essentially linear in $z$. 
     \item }
      The expression:
%
    $  z +
    (1-\signb{x})\cdot (1-\signb{-x})\cdot (y-z) $
    is essentially constant in $x$ and linear in $y$ and $z$.
%

\begin{definition}[Linear length ODE \THEPAPIERS]\label{def:linear lengt ODE}
Function $\tu f$ is linear $\lengt$-ODE definable (from $\tu u$, 
$\tu g$ and $\tu h$) if it corresponds to the
solution of 
%
\begin{equation} \label{SPLode}
f(0,\tu y) 
= \tu g(\tu y)   \quad  and \quad
\dderivl{\tu f(x,\tu y)} 
=   \tu u(\tu f(x,\tu y), \tu h(x,\tu y),
x,\tu y) 
\end{equation}
\noindent where $\tu u$ is \textit{essentially linear} in $\tu f(x, \tu y)$. 
\end{definition}


\section{Concepts from computable analysis}
\label{sec:defanalysecalculable}


When we say that a function $f: S_{1}  \times \dots \times S_{d} \to \R^{d'}$ is (respectively: polynomial-time) computable this will always be in the sense of computable analysis. We recall here the basic concepts and definitions, mostly following the book \cite{Ko91}, whose subject is complexity theory in computable analysis. Alternative presentations include \cite{brattka2008tutorial,Wei00}. 
Actually, as we want to talk about functions in $\Lesfonctionsquoi$, we need to mix complexity issues dealing with integer and real arguments.

\shortermcu{
\begin{remark} One difficulty is that \cite{Ko91} does not formalize all the statements for general functions of $\Lesfonctionsquoi$, and actually almost always restricts to functions over compact domains, and hence we cannot always refer to statements fully formalized in this book, and this why we formulate here explicitly some of the statements and definitions.
\end{remark}
}


\shortermcu{
\subsection{On computable analysis: Computability}}

%
%



A dyadic  number $d$ is a rational number with  a finite binary expansion. That is to say $d=m / 2^{n}$ for some integers $m \in \Z$, $n\in \N$, $n \geq 0$. Let $\dyadic$ be the set of all dyadic rational numbers. We denote by $\dyadic_{n}$ the set of all dyadic rationals $d$ with a representation $s$ of precision $\operatorname{prec}(s)=n$; that is, $\dyadic_{n}=\left\{m \cdot 2^{-n} \mid m \in \Z\right\}$.

\begin{definition}[\cite{Ko91}]  \label{def:cinq} For each real number $x$, a function $\phi: \N \rightarrow \dyadic$ is said to binary converge to $x$ if  for all $n \in \N, \operatorname{prec}(\phi(n))=n$ and $|\phi(n)-x| \leq 2^{-n}$. Let $C F_{x}$ (Cauchy function) denote the set of all functions binary converging to $x$.
\end{definition}

%
%
%
%
%

Intuitively Turing machine $M$ computes a real function $f$ in the following way: 1. The input $x$ to $f$, represented by some $\phi \in C F_{x}$, is given to $M$ as an oracle; 2. The output precision $2^{-n}$ is given in the form of integer $n$  as the input to $M$; 3. The computation of $M$ usually takes two steps, though sometimes these two steps may be repeated for an indefinite number of times:
4. $M$ computes, from the output precision $2^{-n}$, the required input precision $2^{-m}$; 5. $M$ queries the oracle to get $\phi(m)$, such that $\|\phi(m)-x\| \leq 2^{-m}$, and computes from $\phi(m)$ an output $d \in \dyadic$ with $\|d-f(x)\| \leq$ $2^{-n}$.
%
More formally:

\begin{definition}[\cite{Ko91}] A real function $f: \R \rightarrow \R$ is computable if there is a function-oracle {TM} $M$ such that for each $x \in \R$ and each $\phi \in C F_{x}$, the function $\psi$ computed by $M$ with oracle $\phi$ (i.e., $\left.\psi(n)=M^{\phi}(n)\right)$ is in $C F_{f(x)}$. 
\shortermcu{We say the function $f$ is computable on interval $[a, b]$ if the above condition holds for all $x \in[a, b]$.}
\end{definition}

\shortermcu{
\begin{remark}
Given some $x \in \R$, such  an oracle TM $M$ can determine some integer $X$ such that $x \in [-2^{X},2^{X}]$.
\end{remark}
}

\olivier{Probablement superflus pour l'instant:
The following concept plays a very important role:

\begin{definition} \label{def:above}
Let $f:[a, b] \rightarrow \R$ be a continuous function on $[a, b]$. Then, a function $m: \N \rightarrow \N$ is said to be a modulus function of $f$ on $[a, b]$ if for all $n \in \N$ and all $x, y \in[a, b]$, we have
$$
|x-y| \leq 2^{-m(n)} \Rightarrow|f(x)-f(y)| \leq 2^{-n}
$$
\end{definition}

The following is well known (see e.g. \cite{Ko91} for a proof):

\begin{theorem}
A function $f: \R \rightarrow \R$ is computable iff there exist two recursive functions $m: \N \times \N \rightarrow \N$ and $\psi: \dyadic \times \N \rightarrow \dyadic$ such that
\begin{enumerate}
\item for all $k, n \in \N$ and all $x, y \in[-k, k],|x-y| \leq 2^{-m(k, n)}$ implies $|f(x)-f(y)| \leq 2^{-n}$, and
\item  for all $d \in \dyadic$ and all $n \in \N,|\psi(d, n)-f(d)| \leq 2^{-n}$.
\end{enumerate}
\end{theorem}
}

\shortermcu{
\subsection{On computable analysis: Complexity}
}

%
%
%
Assume that $M$ is an oracle machine which computes $f$ on do$\operatorname{main} G$. For any oracle $\phi \in C F_{x}$, with $x \in G$, let $T_{M}(\phi, n)$ be the number of steps for $M$ to halt on input $n$ with oracle $\phi$, and $T_{M}^{\prime}(x, n)=\max \left\{T_{M}(\phi, n) \mid \phi \in C F_{x}\right\}$. The time complexity of $f$ is defined as follows.

\begin{definition}[\cite{Ko91}]
Let $G$ be  bounded closed interval $[a, b]$. Let $f: G \rightarrow \R$ be a computable function. Then, we say that the time complexity of $f$ on $G$ is bounded by a function $t: G \times \N \rightarrow \N$ if there exists an oracle TM $M$ which computes $f$ 
such that for all $x \in G$ and all $n>0$, $T_{M}^{\prime}(x, n) \leq t(x, n)$.
\end{definition}

In other words, the idea is to measure the time complexity of a real function based on two parameters: input real number $x$ and output precision $2^{-n}$. Sometimes, it  is more convenient to simplify the complexity measure to be based on only one parameter, the output precision. For this purpose, we say the uniform time complexity of $f$ on $G$ is bounded by a function $t^{\prime}: \N \rightarrow \N$ if the time complexity of $f$ on $G$ is bounded by a function $t: G \times \N \rightarrow \N$ with the property that for all $x \in G$, $t(x, n) \leq t^{\prime}(n)$.

{
However, if we do so, it is important to realize that if we would have taken $G=\R$ in previous definition, for unbounded functions $f$, the uniform time complexity does not exist, because the number of moves required to write down the integral part of $f(x)$ grows as $x$ approaches $+\infty$ or $-\infty$. Therefore, the approach of \cite{Ko91} is to do as follows (The bounds $-2^{X}$ and $2^{X}$ are somewhat arbitrary, but are  chosen here  because the binary expansion of any $x \in\left(-2^{n}, 2^{n}\right)$ has $n$ bits in the integral part).
%
%
\begin{definition}[Adapted from \cite{Ko91}]  For functions $f(x)$ whose domain is $\R$, 
 we say that the (non-uniform) time complexity of $f$ is bounded by a function $t^{\prime}: \N^{2} \rightarrow \N$ if the time complexity of $f$ on $\left[-2^{X}, 2^{X}\right]$ is bounded by a function $t: \N^{2} \rightarrow \N$ such that $t(x, n) \leq t^{\prime}(X, n)$ for all $x \in\left[-2^{X}, 2^{X}\right]$. 
\end{definition}

\olivier{probablement superflus car on parle pas d'espace:
The space complexity of a real function is defined in a similar way. We say the space complexity of $f: G \rightarrow \R$ is bounded by a function $s: G \times \N \rightarrow \N$ if there is an oracle TM $M$ which computes $f$ such that for any input $n$ and any oracle $\phi \in C F_{x}, M^{\phi}(n)$ uses $\leq s(x, n)$ cells, and the uniform space complexity of $f$ is bounded by $s^{\prime}: \N \rightarrow \N$ if for all $x \in G$ and all $\phi \in C F_{x}, M^{\phi}(n)$ uses $\leq s^{\prime}(n)$ cells.
}

%

As we want to talk about general functions in $\Lesfonctionsquoi$, we extend the approach to more general functions.  
%
%
(for conciseness, when $\tu x=(x_{1},\dots,x_{p})$, $\tu X= (X_{1},\dots, X_{p})$, we write
$\tu x \in [-2^{\tu X}, 2^{\tu X}]$ as a shortcut for $x_{1} \in\left[-2^{X_{1}}, 2^{X_{1}}\right]$,  \dots, $x_{p} \in\left[-2^{X_{p}}, 2^{X_{p}}\right]$). 

\olivier{Pas encore convaincu que c'est propre ce truc.}

\begin{definition}[Complexity for real functions: general case]   \label{def:bonendroit} Consider a function  $f(x_{1},\dots,x_{p},n_{1},\dots,n_{q})$ whose domain is $\R^{p} \times \N^{q}$. 
 We say that the (non-uniform) time complexity of $f$ is bounded by a function $t^{\prime}: \N^{p+q+1} \rightarrow \N$ if the time complexity of $f(\cdot,\dots,\cdot,\ell(n_{1}),\dots,\ell(n_{q}))$ on $\left[-2^{X_{1}}, 2^{X_{1}}\right] \times \dots \left[-2^{X_{p}}, 2^{X_{p}}\right] $  
 is bounded by a function $t(\cdot,\dots,\cdot,\ell(n_{1}),\dots,\ell(n_{q}),\cdot): \N^{p} \times \N \to \N$ such that 
$ t(\tu x,\ell(n_{1}),\dots,\ell(n_{q}), n) \leq t^{\prime}(\tu X,\ell(n_{1}), \dots,\ell(n_{q}), n)$
 whenever  $\tu x \in \left[-2^{\tu X}, 2^{\tu X}\right].$
 We say that $f$ is polynomial time computable if $t^{\prime}$ can be chosen as a polynomial. 
 We say that a vectorial function is polynomial time computable iff all its components are. 
 \end{definition}
 
 \shortermcu{
 \begin{remark}
There is some important {subtlety}: When considering $f: \N \to \Q$, as $\Q \subset \R$, stating $f$ is computable may mean two things: in the classical sense, given integer $y$,  i.e. one can compute $p_y$ and $q_{y}$ some integers such that $f(y)=p_{y}/q_{y}$, or that it is computable in the sense of computable analysis: given some precision $n$,  given arbitrary $y$, and $n$ we can provide some rational (or even dyadic) $q_{n}$ such that $|q_{n}-f(y)| \leq 2^{-n}$. As we said, we always consider the latter.
\end{remark}
}

 We do so that this measures of complexity extends the usual complexity for functions over the integers, where complexity of integers is measured with respects of their lengths, and over the reals, where complexity is measured with respect to their approximation.
 %
 In particular, in the specific case of a function $f: \N^{d} \to \R^{d'}$, that basically means there is some polynomial $t': \N^{d+1} \to \N$ so that the time complexity of producing some dyadic approximating $f(\tu m)$ at precision $2^{-n}$ is bounded by $t'(\ell(m_{1}),\dots,\ell(m_{d}),n)$. 


\olivier{On a besoin de ca? subsection{Some facts from computable analysis} En tous les cas, ca parle de $P_{C}[a, b]$ qu'on a pas introduit. Peut etre faire disparaitre

\begin{theorem}[Alternative characterization \cite{Ko91}]  A function $f$ is in $P_{C}[a, b]$ iff there exist polynomial functions $m$ and $q$ and a function $\psi:(\dyadic \cap[a, b]) \times \N \rightarrow \dyadic$ such that
\begin{enumerate}
\item $m$ is a modulus function for $f$ on $[a, b]$,
\item for any $d \in \dyadic \cap[a, b]$ and all $n \in \N,|\psi(d, n)-f(d)| \leq 2^{-n}$, and
\item $\psi(d, n)$ is computable in time $q(\ell(d)+n)$.
\end{enumerate}
\end{theorem}
}

\olivier{New. Est-ce clair. Est-ce une terminologie élégante. Mieux?}
In other words, when considering that a function is polynomial time computable, it is in the length of all its integer arguments, as this is the usual convention. However, we need sometimes to consider also polynomial dependency directly
	in one of some specific integer argument, say $n_{i}$,  and not on its length $\ell(n_{i})$. We say that \emph{the  function is polynomial time computable, \unaire{n_{i}}} when this holds (keeping possible other integer arguments $n_{j}$, $j \neq i$, measured by their length). 

A well-known observation is the following.

\begin{theorem} Consider $\tu f$ as in Definition \eqref{def:bonendroit} computable in polynomial time. Then $\tu f$ has a polynomial modulus function of continuity, that is to say there is a polynomial function $m_{\tu f}: \N^{p+q+1}\rightarrow \N$ such that for all $\tu x,\tu y$ and all $n>0$, $\|\tu x-\tu y\| \leq 2^{-m_{\tu f}(\tu X,\ell(n_{1}),\dots,\ell(n_{q}), n)}$ implies 
$\|\tu f(\tu x,n_{1}, \dots,n_{q})-\tu f(\tu y,n_{1}, \dots,n_{q})\| \leq 2^{-n}$,
whenever $\tu x,\tu y  \in\left[-2^{\tu X}, 2^{\tu X}\right]$.
\end{theorem}

\olivier{A partir de là, c'est nos résultats}

\section{Functions from $\manonclass$ are in $\FPtime$}
\label{manonclassdansfptime}

The following proposition is proved by induction\footnote{Details on proofs are in Section \ref{sec:proofs} in appendix.} from standard arguments. The hardest  part is to prove that the class of polynomial time computable functions is preserved by the linear length ODE schema:  This is  Lemma \ref{lem:un}.

\begin{proposition} \label{prop:mcu:un}
All functions of $\manonclass$ are computable (in the sense of computable analysis) in polynomial time.
\end{proposition}

\begin{proof}
	This is proved by induction.  This is true for basis functions, from basic arguments from computable analysis. In particular as $\signb{.}$ is a continuous piecewise affine function with rational coefficients, it is computable in polynomial time from standard arguments. 
	
	Now, the class of polynomial time computable functions is  preserved by composition. The idea of the proof for $COMP(f,g)$, is that by induction hypothesis, there exists $M_f$ and $M_g$ two Turing machines computing in polynomial time $f: \RR \rightarrow \RR$ and $g : \RR \rightarrow \RR$. In order to compute $COMP(f,g)(x)$ with precision $2^{-n}$, we just need to compute $g(x)$ with a precision $2^{-m(n)}$, where $m(n)$ is the polynomial modulus of continuity of $f$. 
	%
	Then, we compute $f(g(x))$, which, by definition of $M_f$ takes a polynomial time in $n$. 
	Thus, since $\mathrm{P}_\RR^{\mathrm{P}_\RR} = \mathrm{P}_\RR$, $COMP(f,g)$ is computable in polynomial time, so the class of polynomial time computable functions is preserved under composition.
	It only remains to prove that the class of polynomial time computable functions is preserved by the linear length ODE schema: This is Lemma \ref{lem:un}. 
\end{proof}

\begin{lemma}\label{lem:un}
The class of polynomial time computable functions is preserved by the linear length ODE schema.
\end{lemma}

\newcommand{\vertiii}[1]{{\left\vert\kern-0.25ex\left\vert\kern-0.25ex\left\vert #1 
    \right\vert\kern-0.25ex\right\vert\kern-0.25ex\right\vert}}
 \newcommand\tnorm[1]{\vertiii{#1}}

\olivier{Avis au peuple: Notation $\tnorm{.}$ pourrie? euh, est-ce pas bien de noter comme ca un truc qui n'est pas classique avec une notation qui veut en général dire autre chose.} 

We propose to write $\MYVEC{x}$ for $2^{x}-1$ for conciseness.
We write $\tnorm{\cdots}$ for the sup norm of integer part: given some matrix $\tu
A=(A_{i,j})_{1 \le i \le n, 1 \le j \le m}$, 
$\tnorm{\tu A}=\max_{i,j}
\lceil A_{i,j} \rceil $. In particular, given a vector $\tu x$, it can be seen as a matrix with $m=1$, and $\tnorm{\tu x}$ is the sup norm of the integer part of its components.

\begin{proof} Using   Lemma
  \ref{fundob} in appendix (This lemma is repeated from \THEPAPIERS), when the schema of Definition \ref{def:linear lengt ODE} holds, 
  we can do a change of variable to consider $\tu f(x,\tu y)=\tu F(\ell(x),\tu y)$, with $\tu F$ solution of a discrete ODE of the form
  $\dderiv{\tu F(t,\tu y)}{t} = {\tu A} ( \tu F(t,\tu y), \tu h(\MYVEC{t},\tu y),
	\MYVEC{t}, 
	\tu y) \cdot
	  \tu F(t,\tu y)
	  +   {\tu B} ( \tu F(t,\tu y), \tu h(\MYVEC{t},\tu y),
	  \MYVEC{t}, 	  \tu y),$
	  that is to say, of the form  \eqref{eq:bcg} below. It then follows from:
  \end{proof}

	\begin{lemma}[Fundamental observation] \label{fundamencoreg}
	Consider the ODE 
	\begin{equation} \label{eq:bcg}
	\tu F^\prime(x,\tu y)=  {\tu A} ( \tu F(x,\tu y), \tu h(\MYVEC{x},\tu y),
	\MYVEC{x},
	\tu y) \cdot
	  \tu F(x,\tu y)
	  +   {\tu B} ( \tu F(x,\tu y), \tu h(\MYVEC{x},\tu y),
	  \MYVEC{x},
	  \tu y).
	\end{equation}
	Assume:
1. The initial condition $\tu G(\tu y) = ^{def}
	  \tu F(0, \tu y)$, as well as $\tu h(\MYVEC{x},\tu y)$ 
	  are polynomial time computable \unaire{x}. 
	  2.  ${\tu A} ( \tu F(x,\tu y), \tu h (\MYVEC{x},\tu y),
	  \MYVEC{x},
	  \tu y)$ and ${\tu B} ( \tu F(x,\tu y), \tu h(\MYVEC{x},\tu y),
	  \MYVEC{x},
	  \tu y)$ are \polynomial{} expressions essentially constant in $\tu F(x,\tu y)$.
	
	%
	
	Then, there exists a polynomial $p$ such that $\length{\tnorm{\tu F(x,\tu y)}}\leq p(x,\length{\tnorm{\tu y}})$ and 
	$\tu F(x,\tu y)$ is polynomial time computable \unaire{x}.
	\end{lemma}
\begin{proof}
The fact that there exists a polynomial $p$ such that $\length{\tnorm{\tu f(x,\tu y)}}\leq p(x,\length{\tnorm{\tu y}})$, follows from the fact that we can write some explicit formula for the solution of \eqref{eq:bcg}: This is Lemma \ref{def:solutionexplicitedeuxvariables} in appendix, repeated from  \THEPAPIERS. 
Now,  bounding the size of the right hand side of formula \eqref{eq:rq:fund} provides the statement. 

Now the fact that $\tu F(x,\tu y)$ is polynomial time computable, follows from a reasoning similar to the one of following lemma (the lemma below restricts the form of the recurrence  by lack of space, but the more general recurrence of \eqref{eq:bcg} would basically not lead to any difficulty): The fact that the modulus of continuity of a linear expression of the form of the right hand side of \eqref{eq:bcg} is necessarily affine in its first argument follows from the hypothesis and from previous paragraph, using the fact that $\signb$ has a linear modulus of convergence.
\end{proof}

\begin{lemma}
Suppose that function $\tu f: \N \times \R^{d} \to \R^{d'}$ is such that for all $x, \tu y$,  
\olivier{Pas assez général sous cette forme, meme si ca marche pareil. Tenter d'escroquer en disant qu'on manque de place. Est-ce une escroquerie?}
$$		\tu f(0,\tu y) 
		=\tu g(\tu y) 
		\quad and \quad
		\tu f(x+1,\tu y) 
		= \tu h(\tu f(x,\tu y),x,\tu y)) 
	$$
for some functions $\tu g: \R^{d} \to \R^{d'}$ and $\tu h: \R^{d'} \times \R \times \R^{d} \to \R^{d'}$
both computable in polynomial time \unaire{x}.
%
Suppose that the modulus $m_{h}$ of continuity of $\tu h$ is affine in its first argument:  
	For all $\tu f,\tu f' \in [-2^{\tu F}, 2^{\tu F}]$, $\tu y \in [-2^{\tu Y}, 2^{\tu Y}]$, 
	$\|\tu f-\tu f'\| \le 2^{-m_{h}(\tu F,\ell(x),\tu Y,n)}$ implies $|\tu h(\tu f,x,\tu y)-\tu h(\tu f',x,\tu y)| \le 2^{-n}$
with $m_{h}(\tu F, \ell(x),\tu Y,n) = \alpha n + p_{h}(\tu F,\ell(x),\tu Y)$ for some $\alpha$.
Suppose there exists a polynomial $p$ such that $\length{\tnorm{\tu f(x,\tu y)}}\leq p(x,\length{\tnorm{\tu y}})$.

Then $\tu f(x, \tu y)$ is computable in  polynomial time \unaire{x}.
\end{lemma}

%
%
%
%
%
%
%
%
%
%
%
%
%

\begin{proof}
The point is that we can compute  $\tu f(n,\tu y)$ by
\shortermcu{
\begin{align*}
\tu f(n, \tu y) &= \tu h(\tu f(n-1,\tu y), n-1, \tu y) \\
&= \tu h(\tu h(\tu f(n-2,\tu y), n-2, \tu y), n-1, \tu y) \\
&= \dots \\
&= \underbrace{\tu h(\tu h(\dots \tu h}_{n}(\underbrace{f(0,\tu y)}_{g(\tu y)}, 0, \tu y)\dots), n-1, \tu y)
\end{align*}

Basically, the strategy is to compute 
}
 $\tu z_{0}=\tu f(0,\tu y)=\tu g(\tu y)$,
 then $\tu z_{1}=\tu f(1,\tu y) = \tu h(\tu z_{0}, 0, \tu y)$, 
then $\tu z_{2}=\tu f(2,\tu y) = \tu h(\tu z_{1}, 1, \tu y)$, 
 then \dots,
  then $\tu z_{m}=\tu f(m,\tu y) = \tu h(\tu z_{m-1}, m-1, \tu y)$.
One needs to do so with some sufficient precision so that the result given by $\tu f(l,\tu y)$ is correct, and so that the whole computation can be done in polynomial time. 

Given $\tu y$, we can determine $\tu Y$ such that $\tu y \in [-2^{\tu Y},2^{\tu Y}]$. 
Assume for now that for all $m$,
\begin{equation}
\label{eq:bienborne}
z_{m} \in [-2^{Z_{m}},2^{Z_m}]
\end{equation}

\olivier{Version top-down: pour aider à comprendre.

\begin{itemize}
\item 
It is basically sufficient to determine $z_{l}=\tu f(l,\tu y)=\tu h(z_{l-1},l-1,\tu y)$ with precision $2^{-n}$.  (*)

\item To get such an approximation (*), it suffices to approximate $z_{l-1}$ 
with precision $2^{-m_{h}(Z_{l-1},\ell(l-1),Y,n)}$ (**). 
Then
indeed, $z_{l}$ could then be computed in a time $poly(Z_{l-1},\ell(l-1), Y, n )$.  

\item To get such an approximation (**) of  $z_{l-1}=\tu f(l-1,\tu y)=\tu h(z_{l-2},l-2,\tu y)$ 
, it suffices to approximate $z_{l-2}$
with precision $2^{-m_{h}(Z_{l-2},\ell(l-2),Y,m_{h}(Z_{l-1},\ell(l-1),Y,n))}$ 

We have: 
$$m_{h}(Z_{l-2},\ell(l-2),Y,m_{h}(Z_{l-1},\ell(l-1),Y,n)) = \alpha^{2} n + \alpha p_{h}(Z_{l-1},\ell(l-1),Y) + p_{h}(Z_{l-2},\ell(l-2),Y)$$

Then
indeed, $z_{l-1}$ could then be computed in a time $poly(Z_{l-2},\ell(l-2), Y, m_{h}(Z_{l-1},\ell(l-1),Y,n) )$.  
\item and so on, until $z_{0}$

\end{itemize}
}

\olivier{Version bottom-up}

For $i=0,1,\dots l$, consider 
$p(i)= \alpha^{l-i} n + \sum_{k=i}^{l-1} \alpha^{k-i} p_{h}(\tu Z_{k},\ell(k),\tu Y).$

Using the fact that $\tu g$ is computable, approximate $\tu z_{0}=\tu g(\tu y)$ with precision $2^{-p(0)}$. This is doable polynomial time \unaire{p(0)}.

Then for $i=0,1,\dots, l$, using the approximation of $\tu z_{i}$ with  precision $2^{-p(i)}$, compute an approximation of $\tu z_{i+1}$ with precision $2^{-p(i+1)}$: this is feasible to get precision $2^{-p(i+1)}$ of $\tu z_{i+1}$, as $\tu z_{i+1}=\tu f(i+1,\tu y) = \tu h(\tu z_{i},i,\tu y)$, 
 it is sufficient to consider precision 
 $$
 \begin{array}{lll}
 m_{h}(\tu Z_i,\ell(i),\tu Y,p(i+1)) &=&\alpha p(i+1) + p_{h}(\tu Z_{i},\ell(i),\tu Y)  \\
 &=& \alpha^{l-i-1+1} n \\ &&+  \sum_{k=i+1}^{l-1} \alpha^{k-i-1+1} p_{h}(\tu Z_{k},\ell(k),\tu Y)
 + p_{h}(\tu Z_{i},\ell(i),\tu Y) \\ &=& p(i).
 \end{array}$$
 Observing that $p(l)=n$, we get $z_{l}$ with precision $2^{-n}$.
 All of this is is indeed feasible in polynomial time \unaire{l}, under the condition that all the $Z_{i}$ remain of size polynomial, that is to say, that we have
 indeed \eqref{eq:bienborne}. But this follows from our hypothesis on $\length{\tnorm{\tu f(x,\tu y)}}$.

\end{proof}

\section{Functions from $\FPtime$ are in $\manonclass$}
\label{fptimedansmanonclass}

This section is devoted to prove a kind of reverse implication of Proposition \ref{prop:mcu:un}: For any polynomial time computable function $\tu f: \N^{d} \to \R^{d'}$, we can construct some function $\tilde{\tu f} \in \manonclass$ that simulates the computation of $f$. This basically requires to be able to simulate the computation of a Turing machine using some functions from $\manonclass$.
%
%
%

\newcommand\base{4}
\newcommand\symboleun{1}
\newcommand\symboledeux{3}
\newcommand\encodageconfiguration{\gamma_{config}}
\newcommand\encodagemot{\gamma_{word}}

%
%
%

Consider without loss of generality some Turing machine $$M= (Q, \{0,1\}, q_{init}, \delta, F) $$ using the  symbols $0,\symboleun,\symboledeux$, where $B=0$ is the blank symbol. The reason of the choice of symbols $\symboleun$ and $\symboledeux$ will be made clear latter.  We assume $Q=\{0,1,\dots,|Q|-1\}$.  Let 
$$ \dots  l_{-k} l_{-k+1} \dots l_{-1} l_{0} r_0 r_1 \dots r_n .\dots$$ 
denote the content of the tape of the Turing machine $M$. In this representation, the head is in front of symbol $r_{0}$, and $l_i, r_{i} \in  \{0,\symboleun,\symboledeux\}$ for all $i$. 
Such a configuration $C$ can be denoted by $C=(q,l,r)$, where $l,r \in \Sigma^{\omega}$ are (possibly infinite, if we consider that the tape can be seen as a non finite word, in the case there is no blank on it) words over alphabet $\Sigma=\{\symboleun,\symboledeux\}$ and $q \in Q$ denotes  the internal state of $M$.

The idea is that such a configuration $C$ can also be encoded by some element $\encodageconfiguration(C)=(q, \bar l,\bar r) \in \N \times \R^{2}$, by considering 
\begin{eqnarray*}
  \bar r &=& r_0 \base^{-1} + r_1 \base^{-2} + \dots +  r_n \base^{-(n+1)} + \dots , \\
  \bar l &= & l_{0} \base^{-1} + l_{-1} \base^{-2} + \dots + l_{-k} \base^{-(k+1)} + \dots 
\end{eqnarray*} 

Basically, in other words,  we encode the configuration of bi-infinite tape Turing machine $M$ by real numbers using their radix \base{}  encoding, but using only digits $\symboleun$,$\symboledeux$. 
If we write: $\encodageconfiguration: \Sigma^{\omega} \to \R$ for the function that maps word $w=w_{0} w_{1} w_{2} \dots$ to
$\encodagemot(w)= w_0 \base^{-1} + w_1 \base^{-2} + \dots +  w_n \base^{-(n+1)} + \dots$, we can also write
$\encodageconfiguration(C)=\encodageconfiguration(q,l,r)= (q,\encodagemot(l),\encodagemot(r)).$

\newcommand\Image{\mathcal{I}}

Notice that this lives in $Q \times [0,1]^{2}$. Actually, if we write denote the image of $\encodagemot: \Sigma^{\omega} \to \R$ by $\Image$, this even lives in $Q \times \Image^{2}$. 
%
%
%

\shortermcu{
A key point is to observe that }
\begin{lemma}
We can construct some function $\bar {Next}$ in $\manonclass$ that simulates one step of $M$, i.e. that computes the $Next$ function sending a configuration $C$ of Turing machine $M$ to the next one.  This function is essentially linear.
\end{lemma}

\begin{proof}

We can write  $l = l_0 l^\bullet $ and $r = r_0r^\bullet $, where $l^\bullet$ and $r^\bullet$  corresponding to  (possibly infinite) word $l_{-1} l_{-2} \dots$ and $r_{1} r_{2} \dots$ respectively.
%
%
\vspace{-0.2cm}
	\begin{center}
	\begin{tabular}{c c|c|c|c c}
		\hline 
		... & $l^\bullet $ & $l_0 $ & $r_0$ & $ r^\bullet$ & ... \\ 
		\hline 
		\multicolumn{1}{c}{} & 
		\multicolumn{2}{@{}c@{}}{$\underbrace{\hspace*{\dimexpr6\tabcolsep+3\arrayrulewidth}\hphantom{012}}_{l}$} & 
		\multicolumn{2}{@{}c@{}}{$\underbrace{\hspace*{\dimexpr6\tabcolsep+3\arrayrulewidth}\hphantom{3}}_{r}$}
	\end{tabular} 	
	\end{center}
	
%
%
The function $ {Next}$ is basically of the form
\begin{align*}
\mathit{Next}(q,l,r) &= \mathit{Next}(q,l^\bullet l_0,r_0r^\bullet) = (q', l', r')\\
&= (q', l^\bullet l_0 x, r^\bullet) ~ whenever  ~ \delta(q,r_{0}) = (q',x, \rightarrow) \\
& ~~~~ (q', l^\bullet, l_0 x r^\bullet) ~ whenever ~ \delta(q,r_{0}) = (q',x, \leftarrow) \\
&~~~~~ \dots
\end{align*} \\[-0.7cm]
where the dots is a list of lines of similar types for the various values of $q$ and $r_0$.
This rewrites as a function $\bar{Next}$ which is similar, working over the representation of the configurations as reals:
\begin{align*}
\mathit{\bar{Next}}(q,\bar l, \bar r) &= \mathit{\bar{Next}}(q,\bar{l^\bullet l_0},\bar{r_0r^\bullet}) = (q', \bar{l'}, \bar{r'})\\
&= (q', \bar{l^\bullet l_0 x}, \bar{r^\bullet}) ~ whenever  ~ \delta(q,r_{0}) = (q',x, \rightarrow) \\
& ~~~~ (q', \bar{l^\bullet}, \bar{l_0 x r^\bullet}) ~ whenever ~ \delta(q,r_{0}) = (q',x, \leftarrow) \\
&~~~~~ \dots
\end{align*} \\[-0.7cm]
%
\vspace{-0.5cm}
\begin{equation}
\begin{array}{l} \label{textaremplacer}
\mbox{where $r_{0} =  \lfloor \base \bar{r}\rfloor$ and  \hfill } \\
\mbox{
$\bullet$  in the first case ``$\rightarrow$'' : $\bar{l'} = \base^{-1} \bar l + \base^{-1} x $ and $\bar{r'}  = \bar{r^\bullet} = \{\base \bar r\} $} \\
\mbox{
$\bullet$ in the second case ``$\leftarrow$'' : $\bar{l'} =\bar{ l^\bullet} = \{\base \bar l \} $ and $\bar{r'} = \base^{-2} \bar{r^\bullet} + \base^{-1} x + \lfloor \base \bar{l}\rfloor $ }
\end{array}
\end{equation}

Here $\{.\}$ stands for fractional part.

The problem about such expressions is that we cannot expect the integer part and the fractional part function to be in $\manonclass$ (as functions of this class are computable, and hence continuous, unlike the fractional part).   But, a key point is that from our trick of using only symbols $\symboleun$ and $\symboledeux$, we are sure
that in an expression like $\lfloor \bar r \rfloor$, either it values $0$ (this is the specific case where there remain only blanks in $r$), or that $\base \bar r$ lives in interval
$[\symboleun,\symboleun+1)$ or in interval $[\symboledeux,\symboledeux+1)$. That means that we could replace $\{\base \bar r\}$ by $\sigma(\base \bar r)$ where $\sigma$ is some  (piecewise affine) function obtained by composing in a suitable way the basic functions of $\manonclass$. 
 Namely, 
 define $If (b, T, E)$ as a synonym for $ \signb{b} \times T + (1 - \signb{b}) \times E$.  Then, considering $i(x)= If(x,0,If(x-1,1,3))$, $\sigma(x)=x-i(x)$, then  $i(\base \bar r)$ would be the same as $\lfloor \base \bar r \rfloor$, and $\sigma(\base \bar r)$ would be the same as $\{ \base \bar r \}$ in our context in above expressions.
In other words, we could replace the paragraph \eqref{textaremplacer} above by:
\begin{equation*}
\begin{array}{l}
\mbox{
where $r_{0}= i(\base \bar r)$} \\
\mbox{$\bullet$ in the first case ``$\rightarrow$'' : $\bar{l'} = \base^{-1} \bar l + \base^{-1} x $ and $\bar{r'}  = \bar{r^\bullet} = \sigma(\base \bar r) $} \\
\mbox{$\bullet$ in the second case ``$\leftarrow$'' : $\bar{l'} =\bar{ l^\bullet} = \sigma( \base \bar l) $ and $\bar{r'} = \base^{-2} \bar{r^\bullet} + \base^{-1} x + i(\base \bar{l})$}
\end{array}
\end{equation*}
and get something that would be still work exactly, but using only functions from $\manonclass$. 
Notice that these imbrications of $If$ rewrite to an essentially constant expression.
%
%
%
%

We can then write: 
$$q'= If(q-0, nextq^{0}, If(q-1, nextq^{1}, \cdots, If(q-|Q-2|, nextq^{|Q|-2},nextq^{|Q|-1})))$$
\olivier{
$$q'= \sum_{u=0}^{|Q|-1} \left( \prod_{i=0}^{u-1} Cond(q-u) \right) \cdot (1-Cond(q-i)) \cdot next^{q}$$
}
where 
$$nextq^{q}= If(v-0,nextq^{q}_{0},If(v-\symboleun,nextq^{q}_{\symboleun},nextq^{q}_{\symboledeux}))$$
\olivier{
$$nextq^{q}= \sum_{v \in \{0,\symboleun,\symboledeux\}} \left( \prod_{j \in \{0,\symboleun,\symboledeux\}, j<i}  Cond(r - \sigma(\base \bar r)-v) \right) \cdot (1-Cond(r - \sigma(\base \bar r)-j) \cdot next^{q}_{v}$$
}
and where $nextq^{q}_{v}=q'$ if $\delta(q,v)= (q',x,m)$ for $m \in \{\leftarrow,\rightarrow\}$, for $v \in \{0,\symboleun,\symboledeux\}$. 
Similarly, we can write 
$$r'= If(q-0, nextr^{0}, If(q-1, nextr^{1}, \cdots, If(q-|Q-2|, nextr^{|Q|-2},nextr^{|Q|-1})))$$
\olivier{
$$r'=  \sum_{u=0}^{|Q|-1} \left( \prod_{i=0}^{u-1} Cond(q-u) \right) \cdot (1-Cond(q-i)) \cdot next^{r},$$
}
where 
$nextr^{q}= If(v-0,nextr^{q}_{0},If(v-\symboleun,nextr^{q}_{\symboleun},nextr^{q}_{\symboledeux}))$
\olivier{$$next^{r}= \sum_{v \in \{0,\symboleun,\symboledeux\}} \left( \prod_{j \in \{0,\symboleun,\symboledeux\}, j<i}  Cond(r - \sigma(\base \bar r)-v) \right) \cdot (1-Cond(r - \sigma(\base \bar r)-j) \cdot next^{r}_{v}$$
}
and where $nextr^{q}_{v}$ that corresponds to the corresponding expression in the item above according to the value of $\delta(q,v)$.
We can clearly write a similar expression for $l'$.
These imbrications of $If$ rewrite to  some essentially linear expressions.

\end{proof}

Once we have one step, we can simulate some arbitrary computation of a Turing machine, using some linear length ODE:


\begin{proposition} \label{prop:deux} 
Consider some Turing machine $M$ that computes some function $f: \Sigma^{*} \to \Sigma^{*}$ in some time $T(\ell(\omega))$ on input $\omega$.  One can construct some function $\tilde{\tu f}: \N \times \R \to \R$ in $\manonclass$ that does the same, with respect to the previous encoding: We have $\tilde{\tu f}(2^{T(\ell(\omega))},\encodagemot(\omega))$ provides $f(\omega)$. 
\end{proposition}

\begin{proof}
The idea is to define the function $\bar {Exec}$ that maps some time $2^{t}$ and some initial configuration $C$ to the configuration number at time $t$.  This can be obtained using some linear length ODE using previous Lemma.
$$
 \bar {Exec}(0,C) 
 =C \quad and \quad 
 \dderivl{\bar {Exec}}
  (t, C) 
 =\bar{Next}(\bar {Exec}(t,C)) 
 $$
	
We can then get the value of the computation as $\bar {Exec}(2^{T(\ell(\omega))}, C_{init})$
	on input $\omega$, considering $C_{init}=(q_{0},0,\encodagemot(\omega))$.
	By applying some projection, we get the following function
	$\tilde{\tu f}(x,y)= \projection{3}{3}(\bar {Exec}(x, q_{0},0,y))$ that satisfies the property.	\end{proof}
	
\section{Towards functions from integers to the reals}
\label{sec:computablereal}

The purpose of this section is to prove Theorem \ref{th:main:one}. 
\shortermcu{
\subsection{Reverse implication}
}
The reverse implication of Theorem \ref{th:main:one} mostly follows from Proposition  \ref{prop:mcu:un} and arguments from computable analysis. By lack of space, details are in appendix.

For the direct implication of  Theorem \ref{th:main:one}, the difficulty is that we know from previous section how to simulate Turing machines working over $\Image$,  while we want functions that work directly over the integers and over the reals.  A key is to be able to convert from integers/reals to representations using only symbols $\symboleun$ and $\symboledeux$, that is to say, to map integers to $\Image$, and $\Image$ to reals. 

\begin{lemma}[{From $\Image$ to $\R$}]  \label{lem:codage:manon}
We can construct some  function $\Encode: \N \times [0,1] \to \R$ in $\manonclass$ that  maps $\encodagemot(\overline{d})$ with $\overline{d} \in \{1,3\}^*$  to some real $d$. It is surjective over the dyadic, in the sense that for any dyadic $d \in \dyadic$, there is some (easily computable) such $\overline{d}$ with $\Encode(2^{\ell(\overline{d})},\overline{d})=d$.
\end{lemma}


\begin{proof}
Consider the following transformation: Every digit in the binary expansion of $d$  is encoded by a pair of symbols in the radix $4$ encoding of $\overline{d} \in [0,1]$: digit $0$ (respectively: $1$) is encoded by $11$ (resp. $13$) if before the ``decimal'' point in $d$, and digit $0$ (respectively: $1$) is encoded by $31$ (resp. $33$) if after. For example, for $d=101.1$ in base $2$, $\overline{d}=0.13111333$ in base $4$. 


The transformation from $\overline{d}$ to $d$ can be done by considering a function $F: [0,1]^{2} \to [0,1]^{2}$ that satisfies
$$
F( \overline{r_1}, \overline{l_2}) = 
\left\{
\begin{array}{ll} (\sigma(16 \overline{r_1}), 2 \overline{l_2} + 0) &  \mbox{ whenever } i( 16 \overline{r_1})= 
5\\
(\sigma(16 \overline{r_1}), 2 \overline{l_2} + 1) &  \mbox{ whenever } i( 16 \overline{r_1})= 
7\\
 (\sigma(16 \overline{r_1}), (\overline{l_2} + 0)/2) &  \mbox{ whenever } i( 16 \overline{r_1})= 
 13\\
 (\sigma(16 \overline{r_1}), (\overline{l_2} + 1)/2) &  \mbox{ whenever } i( 16 \overline{r_1})= 
 15
\end{array}\right.
$$
A natural candidate for this is an expression such as 
$If(i(16 \overline{r_1})-0,(\sigma(16 \overline{r_1}), 2 \overline{l_2} + 0),
If(i(16 \overline{r_1})-7,(\sigma(16 \overline{r_1}), 2 \overline{l_2} + 1),
If(i(16 \overline{r_1})-13,(\sigma(16 \overline{r_1}), (\overline{l_2} + 0)/2),
(\sigma(16 \overline{r_1}),$ $(\overline{l_2} + 1)/2))))$
with  $\sigma$ and $i$ constructed as suitable approximation of the fractional and integer part as in previous section.

We provide more details and intuition on the proof of Lemma \ref{lem:codage:manon}:

To compute $d$, given $\overline{d}$, the intuition is to consider a two-tapes Turing machine $(Q, \Sigma, q_{init}, \delta, F)$ : the first tape contains the input ($\overline{d}$), and is read-only, the second one is write-only and empty at the beginning. We just use a different encoding on the second tape that the previous one: For the first tape, we do restrict to digits $0,\symboleun,\symboledeux$, while for the second, we use binary encoding.

Writing the natural Turing machine that does the transformation, this would basically do the  following (in terms of real numbers), if we forget the encoding of the internal state.

Here we write $\overline{ab}$ for the integer whose base $\base$ radix expansion is $ab$. 

This is how we got the function $F$ considered in the main part of the paper. Then the previous reasoning applies on the iterations of function $F$ that would provide some encoding function.

Concerning the missing details on the choice of function $\sigma$ and $i$. From the fact that we have only $\symboleun$ and $\symboledeux$ in $\overline{r}$,  the reasoning is valid as soon as $i(16 \overline{r})$ is correct for $16 \overline{r} \in \{\overline{11},\overline{13},\overline{31},\overline{33})$. So $i(x)=If(x-5,5,If(x-7,7,If(x-13,15))))$ works. Then take $\sigma(x)=x-i(x)$.

We then just need to apply $\ell(\overline{d})$th times $F$ on $(\overline{d},0)$, and then project on the second component to get a function $\Encode$ that does the job. That is $Encode(x,y)= \projection{3}{3}(G(x,y))$
with 
$$
 G(0,y) = (\overline{d},0) 
\quad and \quad 
 \dderivl{G}
  (t, \overline{d},\overline{l}) 
 =F(G(t,\overline{d},\overline{l})).
 $$
\end{proof}

\begin{lemma}[From $\N$ to $\Image$]  \label{lem:manquant} We can construct some function $\Decode: \N^{d} \to \R$ in $\manonclass$ that maps $n \in \N$ to some (easily computable) encoding of $n$ in $\Image$.
\end{lemma}

\begin{proof}
We discuss only the case $d=1$ by lack of space. 
Let $div_{2}$ (respectively: $mod_{2}$)  denote integer (resp. remainder of) division by $2$: As these functions are from $\N \to \N$, from
Theorem \ref{th:ptime characterization 2} from \THEPAPIERS{}, they belongs to $\linearderivlength$.  Their expression in $\linearderivlength$, replacing $\sign$ by $\signb$, provides some extensions $\overline{div_{2}}$ and $\overline{mod{2}}$ in $\manonclass$.
We then do something similar as in the previous lemma but now with  function 
$$
F( \overline{r_1}, \overline{l_2}) = 
\left\{
\begin{array}{ll} 
(\overline{div_{2}}(\overline{r_1}), (\overline{l_2} + 0)/2) &  \mbox{ whenever } \overline{mod_{2}}( \overline{r_1})=0 \\
(\overline{div_{2}}(\overline{r_1}), (\overline{l_2} + 1)/2) &  \mbox{ whenever } \overline{mod_{2}}( \overline{r_1})=1. \\
\end{array}\right.
$$
\end{proof}

We can now  prove the direct direction of  Theorem \ref{th:main:one}: Assume that $\tu f: \N^{d} \to \R^{d'}$ is computable in polynomial time. That means that each of its components are, so we can consider without loss of generality that $d'=1$. We assume also that $d=1$ (otherwise consider either multi-tape Turing machines, or some suitable alternative encoding in $\Encode$).  That means that we know that there is a TM polynomial time computable functions $d: \N^{d+1} \to \{\symboleun,\symboledeux\}^{*}$ so that on $\tu m,n$ it provides the encoding of some dyadic  $\phi(\tu m,n)$ with $\|\phi(\tu m,n)-\tu f(\tu m)\| \le 2^{-n}$ for all $\tu m$. 

From Proposition \ref{prop:deux}, we can construct $\tilde{d}$ 
with $\tilde{d}(2^{p(max(\tu m,n))},\Decode(n,\tu m))=d(\tu m,n)$ for some polynomial $p$ 
corresponding to the time required to compute  $d$. 

Both functions $\length{\tu x}=\length{x_1}+ \ldots
+	\length{x_p}$ and $B(\tu x)=2^{\length{\tu x}\cdot \length{\tu x}}$ are in $\linearderivlength$ (see \THEPAPIERS). It is easily seen that : $\length{\tu x}^c\leq B^{(c)}(\length{\tu x}))$ where $B^{(c)}$ is the $c$-fold composition of function $B$.

Then   $\tilde{\tu f}(\tu m,n)=Encode(\tilde{d}( B^{(c)}(\max(\tu m,n)), \Decode(n,\tu m)))$  provides a solution such that
 $\|\tilde{\tu f}(\tu m,2^{n})-\tu f(\tu m)\| \le 2^{-n}.$

\section{Proving Theorems \ref{th:main:two} and \ref{th:main:twop}}
\label{sec:main:two}

Clearly Theorem \ref{th:main:twop} follows from the case where $d=1$ and $d'=1$ from Theorem \ref{th:main:two}.  Hence, there only remain to prove
Theorem \ref{th:main:two}. The direct direction is immediate from Theorem \ref{th:main:one}. 
For the reverse direction, by induction, the only thing to prove is that the class of functions from to the integers computable in polynomial time is preserved by the operation $\MANONlim$. Take such a function $\tilde{\tu f}$.  By definition, given $\tu m$, we can compute $\tilde{f}(\tu m, 2^n)$ with precision $2^{-n} $ in time polynomial in $n$. This must be by definition of $\MANONlim$ schema some approximation of $\tu f(\tu m)$, and hence $\tu f$ is computable in polynomial time.

%
%
%

%

\section{Generalizations}
\label{sec:generalizations}

\olivier{Mettre ici certaines genéralisations prouvées par Manon}

Recall that 
a function $M : \N \rightarrow \N$ is a modulus of convergence of $g: \N \to \R$, with $g(n)$ converging toward $0$   when $n$ goes to $\infty$,  if and only if for all $i>M(n)$, we have $\| g(i)   \| \le 2^{-n} $.
A function $M :\N \rightarrow \N$ is a uniform modulus of convergence of a sequence $g: \N^{d+1} \to \R$, with $g(\tu m,n)$ converging toward $0$ when $n$ goes to $\infty$  if and only if for all $i>M(n)$, we have $\| g(\tu m,i)  \| \le 2^{-n} $.
%
Intuitively, the modulus of convergence gives the speed of convergence of a sequence.

\begin{definition}[Operation $\MANONlimd$] Given $\tilde{\tu f}:\N^{d+1} \to \R \in \manonclass$, $g: \N^{d+1} \to \R$ such that
for all $\tu m \in \N^{d}$, $n \in \N$,
$\|\tilde{\tu f}(\tu m,2^{n}) - \tu f(\tu m) \| \le g(\tu m, n)$
under the condition that $0 \le g(\tu m, n)$ is decreasing to $0$, with $\| g(\tu m,p(n)) \| \le  2^{-n}$ for some polynomial $p(n)$
then 
$\MANONlimd(\tilde{\tu f},g)$ is the (clearly uniquely defined) corresponding function  $\tu f: \N^{d} \to \R^{e}$.
\end{definition}

\begin{theorem}
We could replace $\MANONlim$ by $\MANONlimd$ in \shortermcu{the statements of} Theorems \ref{th:main:two} and \ref{th:main:twop}.
\end{theorem}

This is equivalent to prove the following,  and observe from the proof that we can replace in above statement ``$g(\tu m,n)$ going to $0$'' by ``decreasing to $0$'', and last condition by
$\| g(\tu m,p(n)) \| \le  2^{-n}$.


\begin{theorem}\label{th:dix} 
$\tu F: \N^{d} \to \R^{d'}$ is computable in polynomial time iff there exists $\tu f: \N^{d+1} \rightarrow \Q^{d'}$, with $\tu f(\tu m,n)$ computable in polynomial time \unaire{n}, 
and $g : \N^{d+1} \rightarrow \Q$ such that
	\begin{itemize}
		\item $\| \tu f(\tu m,n) - \tu F(\tu m) \| \leq g(\tu m,n) $
		\item $0  \le g(\tu m,n)$ and $g(\tu m,n)$ converging to $0$ when $n$ goes to $+\infty$,  
		\item with a uniform  polynomial modulus of convergence $p(n)$.%
%
	\end{itemize}  
\end{theorem}

\begin{proof}
	$\Rightarrow$ If we assume that $\tu F$ is computable in polynomial time, we set $g(\tu m, n) = 2^{-n}$, and we take the identity as uniform modulus of convergence. 
	
	$\Leftarrow$ Given $\tu m$ and $n$, approximate $\tu f(\tu m,p(n+1)) \in \Q$ at precision $2^{-(n+1)}$ by some dyadic rational $q$ and output $q$. This can be done in polynomial time \unaire{n}. 
	We will then have 
	$$\begin{array}{lll}
	\| q - \tu F(\tu m) \| & \le & \| q- \tu f(\tu m,p(n+1)) \| + \| \tu f(\tu m,p(n+1)) - \tu F(\tu m) \| \\
	& \le &  2^{-(n+1)}+ g(\tu m,p(n+1)) \\
	& \le &  2^{-(n+1)}+  2^{-(n+1)} \le 2^{-n}.
	\end{array}$$

\end{proof}

From the proofs we also get a normal form theorem. In particular, 

\begin{theorem}[Normal form theorem]
Any function $f: \N^{d} \to \R^{d'}$ can be obtained from the class $\manonclasslim$ using only one schema $\MANONlim$ (or $\MANONlimd$).
\end{theorem}

\section{Conclusion and future work}
\label{sec:conclusion}

In this article, we characterized the set of functions from the integer to the reals. As we already said, our aim in a future work is to characterize $\FPtime \cap \R^{\R} $ and not only $\FPtime \cap \R^{\N} $. This is clearly a harder task. In particular, a natural approach would be to consider some function $\Encode$ from $\R$ to $\Image$. Unfortunately, 
%
such a function $decode$ is necessarily discontinuous, hence not-computable, hence cannot be in  the class. 
The approach of \emph{mixing} of \cite{BCGH07} might provide a solution, even if the constructions there, based on (classical) continuous ODEs use deeply some closure properties of these functions that are not true for discrete ODEs. 

\newpage
\printbibliography

\newpage

\appendix

\section{Some results from \THEPAPIERS}

\subsection{Some general statements}

In order to be as self-contained as possible, we recall in this section some results and concepts from \THEPAPIERS. All the statements in this section are already present in \THEPAPIERS: We are just repeating them here in case this helps.  We provide some of the proofs, when they are not in the preliminary ArXiv version.  

As said in the introduction:

\begin{definition}[Discrete Derivative] The discrete derivative of
	$\tu f(x)$ is defined as $\Delta \tu f(x)= \tu f(x+1)-\tu
        f(x)$. We will also write $\tu f^\prime$ for
        $\Delta \tu f(x)$ to help readers not familiar with discrete differences to understand
	statements with respect to their classical continuous counterparts. 
	
\end{definition}

Several results from classical derivatives generalize
to the settings of discrete differences: this includes linearity of derivation $(a \cdot
f(x)+ b \cdot g(x))^\prime = a \cdot f^\prime(x) + b \cdot
g^\prime(x)$, formulas for products
and division such as
 $(f(x)\cdot g(x))^\prime =
 f^\prime(x)\cdot g(x+1)+f(x) \cdot g^\prime(x)= f(x+1)  g^\prime(x) +  f^\prime(x)  g(x)$. 
    Notice that, however, there is no simple equivalent of the chain rule. 
 
A fundamental concept is the following:

\begin{definition}[Discrete Integral]
	Given some function $\tu f(x)$, we write $$\dint{a}{b}{\tu f(x)}{x}$$
	as a synonym for $\dint{a}{b}{\tu f(x)}{x}=\sum_{x=a}^{x=b-1}
        \tu f(x)$ with the convention that it takes value $0$ when $a=b$ and
        $\dint{a}{b}{\tu f(x)}{x}=- \dint{b}{a}{\tu f(x)}{x}$ when $a>b$. 
\end{definition}

The telescope formula yields the so-called Fundamental Theorem of
Finite Calculus: 

\begin{theorem}[Fundamental Theorem of Finite Calculus]
	Let $\tu F(x)$ be some function.
	Then,
	$$\dint{a}{b}{\tu F^\prime(x)}{x}= \tu F(b)-\tu F(a).$$
      \end{theorem}

A classical concept in discrete calculus is the one of falling
	power defined as $$x^{\underline{m}}=x\cdot (x-1)\cdot (x-2)\cdots(x-(m-1)).$$
	This notion is  motivated by the fact that it satisfies a derivative formula $
        (x^{\underline{m}})^\prime  = m \cdot x^{\underline{m-1}}$  similar to the classical
        one for powers in the continuous setting.
In a similar spirit, we 
introduce the concept of falling exponential. 


\begin{definition}[Falling exponential]
	Given some function $\tu U(x)$, the expression $\tu U$ to the
	falling exponential $x$,
	denoted by $\fallingexp{\tu U(x)}$, stands
        for  \begin{eqnarray*}
                \fallingexp{\tu U(x)} &=&
                                                                                (1+ \tu U^\prime(x-1)) \cdots
                                        (1+ \tu U^\prime(1)) \cdot (1+ \tu U^\prime(0))   \\
                                         &=&
                   \prod_{t=0}^{t=x-1} (1+ \tu U^\prime(t)),
                    \end{eqnarray*}
	with the convention that $\prod_{0}^{0}=\prod_{0}^{-1}=\tu {id}$, where $\tu
        {id}$ is the identity (sometimes denoted  $1$ hereafter).
\end{definition}

This is motivated by the remarks that 
	$2^x=\fallingexp{x}$, and
        that the discrete
	derivative of a falling exponential is given by
	$$\left(\fallingexp{\tu U(x)}\right )^\prime = \tu U^\prime(x) \cdot
	\fallingexp{\tu U(x)}$$
	%
      for all $x \in \N$.

\begin{lemma}[Derivation of an integral with parameters]  \label{derivationintegral}
   Consider $$\tu F(x) = \dint{a(x)}{b(x)} {\tu f(x,t)}{t}.$$
   Then \begin{eqnarray*}
          \tu F'(x) &=& \dint{a(x)}{b(x)}{  \frac{\partial \tu f}{\partial
                        x} (x,t)}{t} 
                     + \dint{0}{-a^\prime(x)}{\tu f(x+1,a(x+1)+t)}{t} 
+ \dint{0}{b'(x)}{ \tu f(x+1,b(x)+t ) } {t}. 
\end{eqnarray*}

\noindent In particular, when $a(x)=a$ and $b(x)=b$ are constant functions, $\tu F'(x) = \dint{a}{b}{  \frac{\partial \tu f}{\partial
     x} (x,t)}{t},$
     and when $a(x)=a$ and $b(x)=x$,
     $\tu F'(x) = \dint{a}{x}{  \frac{\partial \tu f}{\partial
     x} (x,t)}{t} + \tu f(x+1,x)$.

\end{lemma}

 \begin{proof}
\begin{eqnarray*}
  \tu F'(x) &=& \tu F(x+1) - \tu F(x) \\
  &=& \sum_{t= a(x+1)}^{b(x+1) - 1} \tu f(x+1,t) -
                          \sum_{t= a(x)}^{b(x) - 1} \tu f(x,t)   \\
  &=& \sum_{t= a(x)}^{b(x)-1} \left( \tu  f(x+1,t) - \tu f(x,t)
      \right)    
     +   \sum_{t=a(x+1)} ^{t= a(x)-1} \tu f(x+1,t) 
                      + \sum_{t=b(x)}
      ^{b(x+1)-1} \tu f(x+1,t) \\
  &=& \sum_{t= a(x)}^{b(x)-1}  \frac{\partial \tu f}{\partial
      x} (x,t)  + \sum_{t=a(x+1)} ^{t= a(x)-1} \tu f(x+1,t) + \sum_{t=b(x)}
      ^{b(x+1)-1} \tu f(x+1,t) \\
  &=& \sum_{t= a(x)}^{b(x)-1}   \frac{\partial \tu f}{\partial
      x} (x,t)  + \sum^{t=-a(x+1)+a(x)-1} _{t=0} \tu f(x+1,a(x+1)+t) \\
      &&
  + \sum_{t=0}
      ^{b(x+1)-b(x)-1} \tu  f(x+1,b(x)+t).
\end{eqnarray*}
\end{proof}

\begin{lemma}[Solution of linear ODE
	] \label{def:solutionexplicitedeuxvariables}
	For matrices $\tu A$ and vectors $\tu B$ and $\tu G$,
	the solution of equation $\tu f^\prime(x,\tu y)= \tu A(\tu f(x,\tu y),\tu h(x, \tu y), x,\tu y) \cdot \tu f(x,\tu y)
	+  \tu
	B (\tu f(x,\tu y), \tu h(x, \tu y),  x,\tu y)$  with initial conditions $\tu f(0,\tu y)= \tu G(\tu y)$ is
	\begin{eqnarray*}\label{soluce}
        \tu f(x,\tu y)  &=&
          \left( \fallingexp{\dint{0}{x}{\tu
                             A(\tu f(t,\tu y),\tu h(t, \tu y), t,\tu y)}{t}} \right) \cdot \tu G (\tu
                             y)  \\
           &&
          +
	\dint{0}{x}{ \left(
		\fallingexp{\dint{u+1}{x}{\tu A(\tu f(t,\tu y),\tu h(t, \tu y), t,\tu y)}{t}} \right) \cdot
              \tu B(\tu f(u,\tu y),\tu h(u, \tu y), u,\tu y)} {u}.
          \end{eqnarray*}

        \end{lemma}
                
\begin{proof}
Denoting the right-hand side by $\tu {rhs}(x,\tu y)$, we  have 
$$\begin{array}{ccl} \bar {\tu {rhs}}^{\prime}(x,\tu y)&=& \tu A(\tu f(x,\tu y),\tu h(x, \tu y), x,\tu y)  \cdot  \left( \fallingexp{\dint{0}{x}{\tu
                             A(\tu f(t,\tu y),\tu h(t, \tu y),  t,\tu y)}{t}} \right)   \cdot \tu G (\tu
                             y) \\
                                     && +
                             \dint{0}{x}{ \left(
		\fallingexp{\dint{u+1}{x}{\tu A(\tu f(t,\tu y), \tu h(t, \tu y),  t,\tu y)}{t}} \right)^{\prime} \cdot
              \tu B(\tu f(u,\tu y), \tu h(u, \tu y), u,\tu y)} {u} \\ 
              && +
                             \left(
		\fallingexp{\dint{x+1}{x+1}{\tu A(\tu f(t,\tu y),\tu h(t, \tu y), t,\tu y)}{t}} \right) \cdot
              \tu B(\tu f(x,\tu y), \tu h(x, \tu y), x,\tu y) \\
              &=& \tu A(\tu f(x,\tu y),\tu h(x, \tu y), x,\tu y)  \cdot  \left( \fallingexp{\dint{0}{x}{\tu
                             A(\tu f(t,\tu y),\tu h(t, \tu y),  t,\tu y)}{t}} \right)   \cdot \tu G (\tu
                             y)  \\ && + \  \tu A(\tu f(x,\tu y), \tu h(x, \tu y), x,\tu y) \cdot \\
                             && \dint{0}{x}{ \left(
		\fallingexp{\dint{u+1}{x}{\tu A(\tu f(t,\tu y),\tu h(t, \tu y), t,\tu y)}{t}} \right)  \tu B(\tu f(u,\tu y), \tu h(u, \tu y), u,\tu y)} {u} \\
                    &&
              + \ \tu B(\tu f(x,\tu y),\tu h(x, \tu y), x,\tu y)
              \\
                            &=&  \tu A(\tu f(x,\tu y), \tu h(x, \tu y), x,\tu y) \cdot \tu {rhs}(x,\tu y) 	+  \tu
	B (\tu f(x,\tu y),\tu h(x, \tu y), x,\tu y)
	              \end{array} 
              $$
              where we have used linearity of derivation and definition of falling exponential for the first term, and derivation of an integral (Lemma \ref{derivationintegral}) providing the other terms to get the first equality, and then the definition of falling exponential.
              This proves the property by unicity of solutions of a discrete ODE, observing that $\bar {\tu {rhs}}(0,\tu y)=\tu G(\tu y)$.
\end{proof}

We write also $1$ for the identity. 

        \begin{remark} \label{rq:fund}
          Notice that this can  be rewritten  as 
          \begin{equation} \label{eq:rq:fund} 
\tu f(x,\tu y)=\sum_{u=-1}^{x-1}  \left(
\prod_{t=u+1}^{x-1} (1+\tu A(\tu f(t,\tu y), \tu h(t, \tu y),  t,\tu y)) \right) \cdot  \tu B(\tu f(u,\tu y), \tu h(u, \tu y), u,\tu y)
,
\end{equation}
with the (not so usual) conventions that for any function $\kappa(\cdot)$,  $\prod_{x}^{x-1} \tu \kappa(x) = 1$ and $\tu
B(-1,\tu y)=\tu G(\tu y)$.
Such equivalent expressions both have a clear computational content. They can
be interpreted as an algorithm unrolling
the computation of   $\tu f(x+1,\tu y)$ from the computation of  $\tu
f(x,\tu y), \tu f(x-1,\tu y), \ldots, \tu f(0,\tu y)$.  
        \end{remark}

A fundamental fact is that the derivation with respects to length provides a way to so a kind of change of variables:
\olivier{Le dire mieux}

\begin{lemma}[Alternative view, case of Length ODEs] \label{fundob}
Let 
$f: \N^{p+1}\rightarrow \Z^d$,
$\lengt:\N^{p+1}\rightarrow \Z$  be some functions and assume that \eqref{lode} holds considering  $\lengt(x,\tu y)=\length{x}$.
Then $\tu f(x,\tu y)$ is given by 
$\tu f(x,\tu y)= \tu F(\length{x},\tu y)$
where $\tu F$ is the solution of initial value problem
\begin{eqnarray*}
\tu F(1,\tu y)&=& \tu f(0,\tu y), \\
\dderiv{\tu F(t,\tu y)}{t} &=&  
\tu h(\tu F(t, \tu y),2^{t}-1,\tu y).
\end{eqnarray*}

\end{lemma}

\olivier{Preuve de ca ou pas? Si pas preuve, donner ref}

\olivier{
\subsection{Results for functions over the integers}

The results established in this section were obtained in \THEPAPIERS for functions over the integers, and polynomial time computability was in the classical sense (as only functions over the integers were considered). 

The point is that we need generalization of them in this article, as we consider functions possibly reals, and (polynomial time) computability is in the sense of computable analysis.

We write $\norm{\cdots}$ for the sup norm: given some matrix $\tu
A=(A_{i,j})_{1 \le i \le n, 1 \le j \le m}$, 
$\norm{\tu A}=\max_{i,j}
A_{i,j}$.

\olivier{Nouvelle version d'Arnaud dans journal.tex}

	\begin{lemma}[Fundamental observation] \label{fundamencore}
	Consider the ODE 
	\begin{equation} \label{eq:bc}
	\tu f^\prime(x,\tu y)=  {\tu A} ( \tu f(x,\tu y), \tu h(x,\tu y),
	x,\tu y) \cdot
	  \tu f(x,\tu y)
	  +   {\tu B} ( \tu f(x,\tu y), \tu h(x,\tu y),
	  x,\tu y).
	\end{equation}
	\textbf{over the integers}.

	Assume:
	\begin{enumerate}
	\item The initial condition $\tu G(\tu y) = ^{def}
	  \tu f(0, \tu y)$, as well as $\tu h(x,\tu y)$ are polynomial time computable in $x$ and the length of   $\tu y$. 
	  \item ${\tu A} ( \tu f(x,\tu y), \tu h (x,\tu y),
	  x,\tu y)$ and ${\tu B} ( \tu f(x,\tu y), \tu h(x,\tu y),
	  x,\tu y)$ are \polynomial{} expressions essentially constant in $\tu f(x,\tu y)$.
	
	%
	\end{enumerate}
	
	Then, there exists a polynomial $p$ such that $\length{\tu f(x,\tu y)}\leq p(x,\length{\tu y})$ and $\tu f(x, \tu y)$ is polynomial time computable in $x$ and the length
	of $\tu y$. 
	\end{lemma}

	}

\olivier{Pas envie de dire ca, non?

The previous statements lead to the following:

\begin{lemma}[Key Observation for linear $\lengt$-ODE]~\label{lem:fundamentalobservationlinearlengthODE}
	Assume that $\tu f$ \textbf{over the integers}  is the solution of \eqref{SPLode} and that functions $\tu g, \tu h, \lengt$ and $Jump_\lengt$ are computable in polynomial time. 
	Then, $\tu f$ 
	is computable in polynomial time (\textbf{in the classical sense}).
\end{lemma}

\begin{proof}
 Using the hypotheses, it comes that the number of elements in $Jump_\lengt$ is polynomial in $\length{x}+\length{y}$. We can
  then replace parameter $x$ and derivation in $\lengt(x,\tu y)$ by a
  parameter $t\leq (\length{x}+\length{y})^c$, for some $c$, and derivation in $t$ by Lemma
  \ref{fundob}.
	
	This leads to an ODE of the form:
	$$
	\tu f’(x,\tu y)= \overline {\tu A} ( \tu f(x,\tu y),
	x,\tu y) \cdot
	\tu f(x,\tu y)
	+   \overline {\tu B} ( \tu f(x,\tu y),
	x,\tu y).
	$$
	by setting 
	\begin{eqnarray*}
		\overline {\tu A} ( \tu f(x,\tu y),
		x,\tu y) &=& {\tu A} ( \tu f(x,\tu y), h(x, \tu y),
		x,\tu y) \\
		\overline {\tu B} ( \tu f(x,\tu y),
		x,\tu y) &=&
		{\tu B} ( \tu f(x,\tu y), h(x, \tu y),
		x,\tu y).
	\end{eqnarray*}
	
	But then Lemma \ref{fundamencore} applies, and we get precisely the
	conclusion, observing that the fact that the
	corresponding matrix $\overline {\tu A}$ and vector $\overline {\tu
		B}$ are essentially constant in $\tu f(x, \tu y)$ guarantees
	hypotheses of Lemma \ref{fundamencore}. 
\end{proof}
}

\newpage

\end{document}